\documentclass{amsart}
\usepackage{amsfonts}
\usepackage{amsmath,amscd,lscape}
\usepackage{amsthm}
\usepackage{amssymb}
\usepackage{latexsym}
\usepackage{mathrsfs}
\usepackage{tikz}

\newcommand{\nc}{\newcommand}
\nc{\ben}{\begin{eqnarray}}
\nc{\een}{\end{eqnarray}}

\newcommand{\beqa}{\begin{eqnarray}}
\newcommand{\eeqa}{\end{eqnarray}}
\nc{\Z}{{\bold Z}}

\usepackage[curve,matrix,arrow,color]{xy}

\usepackage{color}

\newtheorem{cor}{Corollary}[section]

\newtheorem{prop}{Proposition}[section]

\newtheorem{example}{Example}
\newtheorem{defn}{Definition}[section]
\newtheorem{thm}{Theorem}

\newtheorem{rem}{Remark}

\newcommand{\cT}{\mathscr{T}}
\newcommand{\cB}{\mathscr{B}}
\newcommand{\cO}{\mathscr{O}}
\newcommand{\CC}{\mathbb{C}}
\newcommand{\ZZ}{\mathbb{Z}}

\newcommand{\cal}{\mathcal}

\newcommand{\oA}{\overline{A}}
\newcommand{\oG}{\overline{G}}
\newcommand{\oK}{\overline{K}}
\newcommand{\oZ}{\overline{Z}}
\newcommand{\oH}{\overline{H}}
\newcommand{\oE}{\overline{E}}
\newcommand{\oF}{\overline{F}}

\newcommand{\E}{\mathbb{E}}
\newcommand{\F}{\mathbb{F}}
\newcommand{\HH}{\mathbb{H}}

\numberwithin{equation}{section}
\begin{document}

\title[FRT presentation and Onsager algebras]{FRT presentation of the Onsager algebras}
\author{Pascal Baseilhac$^{*}$}
\address{$^*$ Laboratoire de Math\'ematiques et Physique Th\'eorique CNRS/UMR 7350,
 F\'ed\'eration Denis Poisson FR2964,
Universit\'e de Tours,
Parc de Grammont, 37200 Tours, 
FRANCE}
\email{baseilha@lmpt.univ-tours.fr}

\author{Samuel Belliard$^{**}$}
\address{$^{**}$ Institut de Physique Th\'eorique, DSM, CEA, URA2306 CNRS Saclay, F-91191 Gif-sur-Yvette, FRANCE}
\email{samuel.belliard@gmail.com}

\author{Nicolas Cramp\'e$^{\dagger,*}$}
\address{$^\dagger$ Laboratoire Charles Coulomb (L2C), UMR 5221 CNRS-Universit\'e de Montpellier, Montpellier, F-France}
\email{nicolas.crampe@umontpellier.fr}

\begin{abstract} 
A presentation {\`a la} Faddeev-Reshetikhin-Takhtajan (FRT) of the Onsager, augmented Onsager and $sl_2$-invariant Onsager algebras is given, using the framework of the non-standard classical Yang-Baxter algebras. Associated current algebras are identified, and generating functions of mutually commuting quantities are obtained. 
\end{abstract}

\maketitle

\vskip -0.5cm


\vspace{0cm}

\vspace{3mm}

\vspace{2mm}

\section{Introduction}
Introduced by L. Onsager in the investigation of the exact solution of the two-dimensional Ising model \cite{Ons44}, the Onsager algebra is an
 infinite dimensional Lie algebra with two known presentations. The original presentation is given in terms of generators 
$\{A_n,G_m|n,m\in {\mathbb Z}\}$ and relations (see Definition \ref{def:OA}).  A second presentation\footnote{The connection between the Dolan-Grady construction and the original Onsager algebra was understood by J.H.H. Perk in 1982 \cite{Perk1}.} is given in terms
of two generators $A_0,A_1$ satisfying the so-called Dolan-Grady relations (\ref{eq:DG}) \cite{DG82}.  Later on, it was shown that the Onsager algebra 
is isomorphic
to a fixed point subalgebra of the affine Kac-Moody algebra $\widehat{sl_2}$ under the Chevalley involution \cite{Davies,Roan}.\vspace{1mm}

In the context of  mathematics and quantum integrable systems, a $q-$deformed analog of the Onsager algebra has been introduced in recent years \cite{Ter03,B1}. 
It is isomorphic
 to a coideal  subalgebra of $U_q(\widehat{sl_2})$ \cite{BB} (see also \cite{Kolb}). Similarly to the  classical
 (undeformed) Onsager algebra, two presentations are known, which can be viewed as $q-$deformed analog of the original presentations of Onsager \cite{BK,BK4} 
 and Dolan-Grady \cite{Ter03,B1}. Other 
types of algebras have been later on considered in the literature:
 the augmented ($q-$)Onsager algebra \cite{IT,BB3,BC13} and the $U_q(gl_2)-$invariant $q-$Onsager algebra \cite{BB17}. Given the known automorphisms of $U_q(\widehat{sl_2})$, these other types of $q-$Onsager algebras
 can be understood as different coideal subalgebras of $U_q(\widehat{sl_2})$. For all these different types of $q-$Onsager algebras, a third presentation 
 is also identified \cite{B1,BSh1,BB5} using the framework of the quantum reflection algebra 
\cite{Skly88}. Based on this presentation, one obtains current algebras for the $q-$Onsager algebras \cite{BSh1} that 
found applications in the solution of the open XXZ spin  chain \cite{BB3}. This third presentation gives also a tool to derive explicit examples of tridiagonal pairs \cite{Ter03} or to conjecture
a Poincar\'e-Birkhoff-Witt basis for the $q-$Onsager algebra \cite{BB5}, besides the so-called zig-zag basis of Ito-Terwilliger \cite{IT}. \vspace{1mm}

By analogy with the family of $q-$Onsager algebras that admit a presentation within the framework of the quantum reflection algebra, it is
 thus natural to search for this `missing' presentation of the classical Onsager algebras. In this letter, we answer this question. Namely, it is shown that all classical Onsager
 algebras admit a presentation within the framework of the non-standard classical Yang-Baxter algebras. This link offers a new perspective 
 to study the related integrable systems. \vspace{1mm}

This letter is organized as follows. In Section 2, the framework of the classical Yang-Baxter algebra $\cT$ (with central extension)
and the non-standard classical Yang-Baxter algebra $\cB$ is settled. Automorphisms of $\cT$ are considered, which are 
used to construct homomorphisms from $\cB$ to $\cT$. Also, based on solutions of the 
classical reflection equation, a different homomorphism is proposed. Then, two commutative subalgebras of $\cB$ are identified.  In Section 3, we construct  some of the simplest
explicit examples of non-standard classical Yang-Baxter algebras $\cB$ related with the affine Kac-Moody algebra $\widehat{sl_2}$.
First, we introduce some necessary material. Different known presentations
of $\widehat{sl_2}$ are also recalled: the Serre-Chevalley presentation \cite{Kac}, 
the Cartan-Weyl presentation \cite{GO} and the so-called FRT presentation (in honour of
Faddeev-Reshetikhin-Takthadjan \cite{FRT87}).
As an application of the results of Section 2, it is shown that known automorphisms of $\widehat{sl_2}$ are easily recovered from the automorphisms of the classical Yang-Baxter algebra.  Then, three different non-standard classical Yang-Baxter algebras are considered. In each case, the map to subalgebras of $\widehat{sl_2}$ is described in details. 
In Section 4, it is shown that these three non-standard classical Yang-Baxter algebras are isomorphic to the Onsager algebra, the augmented Onsager algebra 
and the $sl_2$ invariant Onsager algebra, respectively. In each case, our construction provides an FRT presentation. 
Their corresponding current presentations are  displayed. 
Finally, as an application of the FRT presentations, we derive generating functions of mutually commuting quantities for each of the Onsager algebras. 
Concluding remarks are given in Section 5.\vspace{1mm}

\section{Classical Yang-Baxter algebras and commutative subalgebras}
In this section, we introduce the basic definitions of the standard and non-standard classical Yang-baxter algebras, denoted $\cT$ and $\cB$, respectively. 
See Definitions \ref{def:T} and \ref{def:B}. A class of automorphisms of  $\cT$ is exhibited in Proposition \ref{prop:tw}, and the 
homomorphic image of $\cB$ into the fixed point subalgebra of $\cT$ under this automorphism is described in Proposition \ref{propB}. 
A different homomorphic image of $\cB$ into a subalgebra of $\cT$ is also given in Proposition \ref{propB2}, based on solutions of the classical reflection equation. 
Finally,  two different commutative subalgebras of $\cB$ are displayed, see Propositions \ref{pr:a1} and \ref{pr:a2}.

\subsection{Classical Yang-Baxter algebras\label{sec:FRT}}


The classical Yang-Baxter equation appeared as a limit of the quantum inverse scattering method \cite{Sk2}
and became quickly popular \cite{BD,dri83,Sem}. In this letter, we use two slightly different definitions (see Definitions \ref{def:cybe} and \ref{def:nscybe}).

\begin{defn}  \label{def:cybe}
The matrix
  $r(x)\in End(\CC^N\otimes \CC^N)$ is called a classical r-matrix if it satisfies the skew-symmetric condition $r_{12}(x)=-r_{21}(1/x)$ and
the classical Yang-Baxter equation
 \begin{equation}\label{eq:CYBE}
  [\ r_{13}(x_1/x_3)\ , \ r_{23}(x_2/x_3)\ ]=[\ r_{13}(x_1/x_3)+  r_{23}(x_2/x_3)\ , \ r_{12}(x_1/x_2)\ ]\;
 \end{equation}
for any $x_1,x_2,x_3$.  Here, $r_{12}(x) = r(x)\otimes I\!\! I$ , $r_{23}(x) =I\!\! I\otimes  r(x) $  and so on.
 \end{defn}
 \vspace{1mm}

For a given skew-symmetric r-matrix $r(x)$ that satisfies the classical Yang-Baxter equation, we now introduce the Lie algebra 
$\cT$ with central extension $c$. Generically, this algebra is called the classical Yang-Baxter algebra. 
Note that this algebra can be understood as a classical analog of the algebra introduced in \cite{RS}.
\vspace{1mm}

Let $\{ E_{ij}\ | \ 1\leq i,j\leq N \}$ 
denotes the standard basis of $End(\CC^N)$ (\textit{i.e.} the $N\times N$ matrices with components $(E_{ij})_{kl}=\delta_{ik}\delta_{jl}$). Below we use the Einstein summation convention by omitting the sums over the repeated indices $i,j$.

\begin{defn}\label{def:T} Let  $r(x)\in End(\CC^N\otimes \CC^N)$ be a classical r-matrix.
 $\cT$ is the Lie algebra with generators\footnote{According to the examples, the range of $\ell$ will be restricted.}
$\{c, t_{ij}^{\pm [\ell]} \ |\ 1\leq i,j\leq N,\ell\in\ZZ \}$. Introduce:
\begin{equation}
 T^+(x)= E_{ij} \otimes  \sum_{\ell\in \ZZ}t_{ij}^{+[\ell]} x^\ell\ ,\quad  T^-(x)= E_{ij} \otimes  \sum_{\ell\in \ZZ}t_{ij}^{-[\ell]} x^\ell.
\end{equation}
The defining relations are:
\begin{eqnarray}
\null[ T^{\pm}(x),c]&=&0\;,\label{Tg}\\
\null[ T_1^{\pm}(x),T^{\pm}_2(y)]&=&[ T_1^{\pm}(x)+T^{\pm}_2(y),r_{12}(x/y)]\;,\label{rpp}\\
\null[ T_1^{+}(x),T^{-}_2(y)]
&=&[ T_1^{+}(x)+T^{-}_2(y),r_{12}(x/y)]-2c\,r'_{12}(x/y)x/y\ .\label{rpm}
\;
\end{eqnarray}
\end{defn}
\begin{rem} The Jacobi identity of the Yang-Baxter algebra $\cT$ follows from the classical Yang-Baxter equation and its derivative. The latter reads:
\begin{eqnarray}
 [ \ f_{13}(x_1/x_3) +  f_{23}(x_2/x_3) \ , \ r_{12}(x_1/x_2) \ ]=[ \ f_{13}(x_1/x_3)\ ,\  r_{23}(x_2/x_3)\ ] +   [  \  r_{13}(x_1/x_3)\ ,\ f_{23}(x_2/x_3)\ ]\;\nonumber
\end{eqnarray}
where $f(x)=xr'(x)$.
\end{rem}
\vspace{1mm}

The presentation of the commutation relations \eqref{Tg}-\eqref{rpm} of a Lie algebra $\cT$ is called a {\it FRT presentation} in honour of the authors  Faddeev-Reshetikhin-Takhtajan \cite{FRT87}. \vspace{1mm}

A class of automorphisms of the Lie algebra $\cT$ is now considered.
\begin{prop}\label{prop:tw} Let $r(x)$ be a classical r-matrix with the additional property
 $ r_{12}(x)=-r_{12}(1/x)^{t_1t_2}$\; and $U(x)$ be a $N\times N$ invertible matrix satisfying
\begin{eqnarray} 
&& U(x)^t=\epsilon \ U(1/x)\quad\text{and}\quad U_1(x)U_2(y) r_{12}(x/y)=r_{12}(x/y)U_1(x)U_2(y)\; \quad \mbox{for $\epsilon=+1$ or $\epsilon=-1$}.\label{eqU}
\end{eqnarray}
Then, the map $\theta:End(\CC^N)\otimes\cT \rightarrow End(\CC^N)\otimes\cT$ such that 
\begin{eqnarray}\label{eq:auU}
 T^\pm(x)\mapsto  U(x)\ T^\mp(1/x)^t\ U(x)^{-1} \ \mp \ c x U'(x)U(x)^{-1}, \quad c \mapsto -c\ ,
\end{eqnarray}
where $(.)^t$ stands for the transposition in $End(\CC^N)$, provides an involutive automorphism of $\cT$. 
\end{prop}
\begin{proof} 
First, we show that $\theta$ is an automorphism. 
The action of $\theta$ on \eqref{Tg} is trivially computed. Consider the action of $\theta$ on (\ref{rpp}). After straightforward simplifications, the corresponding equation  reduces to:
\begin{eqnarray}
 [\ xU_1'(x)U_1(x)^{-1} + yU_2'(y)U_2(y)^{-1}\ ,\ r_{12}(x/y)]=0.\label{eq:cU}
\end{eqnarray}
Similarly, consider the action of $\theta$ on (\ref{rpm}). The corresponding equation reduces to:
\begin{eqnarray} 
&&  2\frac{x}{y}\big( U_1(x)U_2(y)r'_{12}(x/y)U_1(x)^{-1}U_2(y)^{-1} 
-  r'_{12}(x/y) \big) \nonumber\\
&&+ \big[\ xU_1'(x)U_1(x)^{-1} - yU_2'(y)U_2(y)^{-1}\ ,\ r_{12}(x/y)\big] =0.\nonumber
\end{eqnarray}
It is easy to show that these last two equations follow from linear combinations of the derivatives of  (\ref{eqU}) either with respect to $xd/dx$ or with respect to $yd/dy$. It remains to show the involution property. One has:
\begin{eqnarray} 
\theta^2(T^+(x))&=&U(x)(\theta(T^-(1/x)))^tU(x)^{-1} + cxU'(x)U(x)^{-1}\nonumber\\
&=& U(x)(U(1/x)T^+(x)^tU(1/x)^{-1} )^tU(x)^{-1} + cxU'(x)U(x)^{-1}\nonumber\\
&& + c\frac{1}{x}U(x)(U'(1/x)U^{-1}(1/x))^tU(x)^{-1}.\nonumber
\end{eqnarray} 
Then, we use the first equation of (\ref{eqU}) and its derivative with respect to $xd/dx$, which gives $U'(x)=-\epsilon U'(1/x)^t/x^2$. The equation above 
simplifies to $\theta^2(T^+(x))=T^+(x)$. We show similarly $\theta^2(T^-(x))=T^-(x)$ (and $\theta^2(c)=c$), which concludes the proof.
\end{proof}
\vspace{1mm}

 \subsection{Non-standard classical Yang-Baxter equation and the algebra $\cB$ \label{sec:nsFRT}}

The non-standard classical Yang-Baxter equation can be understood as a generalization of the classical Yang-Baxter equation (\ref{eq:CYBE}).
\begin{defn}\label{def:nscybe} The matrix $\overline{r}(x,y)\in End(\CC^N\otimes \CC^N)$ is called a non-standard classical r-matrix  if it is a solution of 
the non-standard classical Yang-Baxter equation
 \begin{equation}\label{eq:nsCYBE}
  [\ \overline{r}_{13}(x_1,x_3)\ , \ \overline{r}_{23}(x_2,x_3)\ ]=[\ \overline{r}_{21}(x_2,x_1)\ , \ \overline{r}_{13}(x_1,x_3)\ ]+[\ \overline{r}_{23}(x_2,x_3)\ , \ \overline{r}_{12}(x_1,x_2)\ ]\;
 \end{equation}
for any $x_1,x_2,x_3$. 
\end{defn}

In the case where the non-standard classical r-matrix 
 depends on only one parameter $\overline{r}_{12}(x,y)=r_{12}(x/y)$
and is skew-symmetric  $\overline{r}_{12}(x,y)=-\overline{r}_{21}(y,x)$
then the non-standard classical Yang-Baxter equation reduces to the standard one (\ref{eq:CYBE}). \vspace{1mm}
 
For a given r-matrix $\overline{r}(x,y)$ that satisfies the non-standard classical Yang-Baxter equation, we now introduce the Lie algebra $\cB$. 
\begin{defn}\label{def:B}  $\cB$ is the Lie algebra with generators 
$\{b_{ij}^{[\ell]} \ |\  \ 1\leq i,j\leq N, \ell\in\ZZ \}$. Introduce:
\begin{equation}
 B(x)= E_{ij}\otimes \sum_{\ell \in\ZZ} b_{ij}^{[\ell]} x^\ell\ .
\end{equation}
The defining relations are:
\begin{equation}\label{eq:Al}
 [\ B_{1}(x)\ , \ B_{2}(y)\ ]=[\ \overline{r}_{21}(y,x)\ , \ B_{1}(x)\ ]+[\ B_{2}(y)\ , \ \overline{r}_{12}(x,y)\ ]\; ,
\end{equation}
where $\overline{r}(x,y)$ satisfies the non-standard classical Yang-Baxter equation (\ref{eq:nsCYBE}).
\end{defn}
Note that the Jacobi identity for $\cB$ is guaranteed by the fact that the classical r-matrix in \eqref{eq:Al} satisfies 
the non-standard classical Yang-Baxter equation (\ref{eq:nsCYBE}). Although this r-matrix is not in general skew-symmetric 
the Lie bracket of $\cB$ is still anticommutative.\vspace{1mm} 

Given an automorphism $\theta$ of $\cT$, a realization of the algebra $\cB$ into a fixed point subalgebra\footnote{Also called twisted algebra, where the twist is the automorphism.} 
of $\cT$
 is now considered. 

%
\begin{prop}\label{propB} Assume $ r_{12}(x)=-r_{12}(1/x)^{t_1t_2}$, $U(x)$ is a solution of (\ref{eqU}) and $\theta$ is defined by \eqref{eq:auU}.
The map $\psi_\theta:\cB \rightarrow \cT$ such that
\begin{equation}\label{eq:Bxg}
 B(x) \mapsto T^+(x) + \theta(T^+(x))=T^+(x) + U(x)\ T^-(1/x)^t\ U(x)^{-1} -c x U'(x) U(x)^{-1}\;,
\end{equation}
with
\begin{equation}\label{eq:rb}
 \overline{r}_{12}(x,y)=r_{12}(x/y)+U_1(x)\ r_{12}^{t_1}(1/(xy))\ U_1(x)^{-1}\;
\end{equation}
is an algebra homomorphism.
\end{prop}
\begin{proof}
See the proof of the Proposition \ref{propB2} which is more general according to the Remark \ref{rem:UK}.
\end{proof}
In the following, we denote by $\cT^\theta \subset \cT$ the  image of the algebra $\mathcal{B}$ by the homomorphism
$\psi_\theta$.
Let us remark that in the previous proposition, we may have considered $B(x) \mapsto T^-(x)+\theta(T^-(x))$. However, it defines an isomorphic subalgebra.\\

To conclude this Section, we would like to point out that other subalgebras of $\cT$, that are not necessarily fixed point subalgebras, may be considered as well. 
Indeed, along the lines of \cite{Skry}, let us consider a special case of the classical reflection equation \cite[eq. (2)]{Sk1}.
\begin{defn}
The matrix $k(x)\in End(\CC^N)$ is called a k-matrix if it is a solution of the following equation
 \begin{equation}\label{eq:re}
  r_{12}(x/y) k_1(x) k_2(y)-k_1(x) k_2(y)r_{12}(x/y)=k_1(x) r_{12}^{t_2}(xy) k_2(y)- k_2(y)r_{12}^{t_2}(xy)k_1(x)
 \end{equation}
for a given skew-symmetric r-matrix $r(x)$ satisfying $ r_{12}(x)=-r_{12}(1/x)^{t_1t_2}$.
\end{defn}

\begin{rem}\label{rem:UK} The  matrices $U(x)$ solving \eqref{eqU} are examples of k-matrices. 
\end{rem}
We are now in position to give a generalization of Proposition \ref{propB}.

\begin{prop}\label{propB2} For a given k-matrix $k(x)$, the map $\cB \rightarrow \cT$ such that
\begin{equation}\label{eq:Bxgg}
 B(x) \mapsto T^+(x) + k(x)\ T^-(1/x)^t\ k(x)^{-1} -c x k'(x) k(x)^{-1}\;,
\end{equation}
with
\begin{equation}\label{eq:rb2}
 \overline{r}_{12}(x,y)=r_{12}(x/y)+k_1(x)\ r_{12}^{t_1}(1/(xy))\ k_1(x)^{-1}\;
\end{equation}
is an algebra homomorphism.
\end{prop}

\begin{proof} Inserting (\ref{eq:Bxgg}) into (\ref{eq:Al}) and using (\ref{rpp}), (\ref{rpm}), one identifies $\overline{r}_{12}(x,y)$ as (\ref{eq:rb2}). Also, one finds a term in $c$ with coefficient:
\begin{eqnarray} 
  2xy\left(k_1(x)r_{12}'^{t_2}(xy)k_2(y)-k_2(y)r_{12}'^{t_2}(xy)k_1(x)\right) &-& x \big[  r_{12}(x/y) + k_2(y)r_{12}^{t_2}(xy)k_2(y)^{-1}  \ , \  k'_1(x)k_1(x)^{-1}\big]k_1(x)k_2(y)\nonumber\\
&+& y \big[ k'_2(x)k_2(x)^{-1}\  ,\  r_{12}(x/y) - k_1(x)r_{12}^{t_2}(xy)k_1(y)^{-1}  \big]k_1(x)k_2(y) .\label{eq:2xy}
\end{eqnarray} 
To show that this coefficient vanishs, we use the reflection equation (\ref{eq:re}) from which we deduce:
\begin{eqnarray} 
&& k_1(x)^{-1}r_{12}(x/y)k_1(x)k_2(y) + k_1(x)^{-1}k_2(y) r_{12}^{t_2}(xy) k_1(x)= r_{12}^{t_2}(xy)k_2(y) + k_2(y)r_{12}(x/y),
\nonumber\\
&& k_2(y)^{-1}r_{12}(x/y)k_1(x)k_2(y) - k_1(x)k_2(y)^{-1} r_{12}^{t_2}(xy) k_2(y)= k_1(x)r_{12}(x/y) - r_{12}^{t_2}(xy)k_1(x).\nonumber
\end{eqnarray} 
Using these two equations, the condition on the coefficient (\ref{eq:2xy})  yields to the constraint:
\begin{eqnarray} 
 [r_{12}(x/y), xk'_1(x)k_2(y) + yk_1(x)k'_2(y)] &-& xk'_1(x)r_{12}^{t_2}(xy)k_2(y)  -yk_1(x)r_{12}^{t_2}(xy)k'_2(y) \nonumber\\
&+&  xk_2(y)r_{12}^{t_2}(xy)k'_1(x) + yk'_2(y)r_{12}^{t_2}(xy)k_1(x)\nonumber\\
&-&  2xy\left(k_1(x)r_{12}'^{t_2}(xy)k_2(y) -k_2(y)r_{12}'^{t_2}(xy)k_1(x)\right) =0.\nonumber 
\end{eqnarray}
Taking linear combinations of the derivatives of  (\ref{eq:re}) either with respect to $xd/dx$ or with respect to $yd/dy$, one finds
that this constraint is satisfied.
\end{proof}

Note that the r-matrix defined by \eqref{eq:rb2} is a solution of the non-standard classical Yang-Baxter equation \eqref{eq:nsCYBE}, as shown in \cite{Skryybe}.\vspace{1mm}

The homomorphic image of $\cB$ defined in Proposition \ref{propB2} is a Lie subalgebra of $\cT$ that is not necessarily a fixed point subalgebra. It is denoted by $\cT^k$.\vspace{1mm}

\subsection{Commutative subalgebras\label{sec:CS}}
Two commutative subalgebras of $\cB$ are now identified. The first subalgebra has already appeared in the literature. The second  subalgebra, which is new, 
will be of interest in the analysis of further sections.

The following proposition has been proven in \cite{Jur,hwa,Hik} to prove the integrability of the Gaudin models \cite{gau} and has been generalized in \cite{Tal}.
\begin{prop}\label{pr:a1}
 The following generating function in the universal enveloping algebra of $\cB$
 \begin{equation}\label{tB2}
  t(x)= tr \left(B(x)^2\right)=\sum_{\ell,p \in \ZZ}\ b_{ij}^{[\ell]}b_{ji}^{[p]}\  x^{\ell+p}\;
 \end{equation}
 satisfies
 \begin{equation}
  [ t(x)\ ,\ t(y) ]=0\;.
 \end{equation}
\end{prop}

\begin{prop}\label{pr:a2} Let  $\overline{r}(x,y)$ be a solution of the non-standard classical Yang-Baxter equation (\ref{eq:nsCYBE}).
 Let $M(x)\in End(\CC^N)$ be a solution of the following equation
 \begin{equation}\label{eq:reD}
  [tr_1 ( \overline{r}_{12}(x,y) M_1(x) ) \ ,\ M_2(y) ]=0\; .
 \end{equation}
Then,
\begin{equation}\label{tb2}
 b(x)=tr M(x) B(x)
\end{equation}
satisfies
 \begin{equation}
  [ b(x)\ ,\ b(y) ]=0\;.
 \end{equation}
\end{prop}
\begin{proof}
Multiply \eqref{eq:Al} on the left by $M_1(x)M_2(y)$  and take the traces in the spaces 1 and 2 to get
\begin{eqnarray}
 [ b(x)\ ,\ b(y) ]&=&tr_{12}M_1(x)M_2(y)\big( \overline{r}_{21}(y,x) B_{1}(x)-B_{1}(x)\overline{r}_{21}(y,x) \big)\nonumber\\
 &&+ tr_{12}M_1(x)M_2(y)\big(B_{2}(y)\overline{r}_{12}(x,y)-\overline{r}_{12}(x,y)B_{2}(y)\big)\nonumber\\
 &=& tr_{1}M_1(x)tr_2\Big(M_2(y) \overline{r}_{21}(y,x)\Big) B_{1}(x)-tr_{1}tr_2\Big(\overline{r}_{21}(y,x)M_2(y)\Big)M_1(x)B_{1}(x) \big)\nonumber\\
 &&+tr_2 tr_{1}\Big(\overline{r}_{12}(x,y)M_1(x)\Big) M_2(y)B_{2}(y)-tr_{2}M_2(y)tr_1\Big(M_1(x)\overline{r}_{12}(x,y)\Big)B_{2}(y)\nonumber\\
 &=& tr_{1}\Big[M_1(x),tr_2\Big(M_2(y) \overline{r}_{21}(y,x)\Big)\Big] B_{1}(x)+tr_2 \Big[tr_{1}\Big(\overline{r}_{12}(x,y)M_1(x)\Big), M_2(y)\Big]B_{2}(y)\;.\nonumber
\end{eqnarray}
The cyclicity of the traces is used to prove the previous relations. Then, by using \eqref{eq:reD}, one concludes the proof.
\end{proof}

\begin{rem} Note that $t(x)$ and $b(y)$ do not necessarily commute.
\end{rem}

\vspace{1mm}
\section{FRT presentation of fixed point subalgebras of  $\widehat{sl_2}$ }
In this section, based on one of the simplest example of classical Yang-Baxter algebra $\cT$  associated with  $\widehat{sl_2}$,
we present four different explicit examples of non-standard  classical Yang-Baxter algebra $\cB$. 
Firstly, the Serre-Chevalley and Cartan-Weyl presentations of the affine Kac-Moody algebra $\widehat{sl_2}$ are recalled. 
Secondly, the FRT presentation of  $\widehat{sl_2}$  is given  and the automorphisms of Proposition \ref{prop:tw} are described, see Corollary \ref{cor:twloop}. Finally, using  Propositions \ref{propB} and \ref{propB2} we  obtain the FRT presentation for four different homomorphic images of $\cB$. They are identified with certain fixed point subalgebras of $\widehat{sl_2}$.

\subsection{Presentations of $\widehat{sl_2}$  \label{sec:L2}}

The affine Kac-Moody algebra $\widehat{sl_2}$  has been studied extensively in the mathematics and physics literature  \cite{Kac,GO} (see also \cite{Lec}), 
where two presentations are usually considered. The so-called Serre-Chevalley presentation is generated by $\{x^+_i,x^-_i ,k_i|i=0,1\}$ subject 
to the commutation relations 
\ben
&&\null[k_i,x^\pm_j]=\pm a_{ij} x^\pm_j, \quad [k_i,k_j]= 0, \quad [x^+_i,x^-_j]=k_{j}\delta_{i+j},\\
&&\null[x^\pm_i,[x^\pm_i,[x^\pm_i,x^\pm_j]]]=0 \quad \mbox{for}\quad i\neq j,
\een
where $\delta_i=0$ if $i\neq 0$ and $\delta_0=1$. 
The element $c=k_0+k_1$ is central and the entries of the Cartan matrix are given by $a_{00}=a_{11}=2$ and $a_{01}=a_{10}=-2$. Alternatively, 
the Cartan-Weyl presentation of $\widehat{sl_2}$  is generated by $\{e_n,f_n ,h_n,c|n \in \ZZ\}$  subject to the commutation relations
\begin{alignat}{4}
 &[e_n,e_m]=[f_n,f_m]=0\ ,\label{CW1}\\
&\null[h_n,e_m]=2 e_{n+m}\ , \qquad && [h_n,f_m]=-2 f_{n+m},\label{CW2}\\
&\null[h_n,h_m]= 2\,c\,n\,\delta_{n+m}\ , \quad &&[e_n,f_m]=h_{n+m}+c\, n\, \delta_{n+m},   \label{CW3}
\end{alignat}
where $c$ is central.  Recall that the isomorphism from the Serre-Chevalley to the Cartan-Weyl presentation is given, up to automorphism, by (see e.g. \cite{Lec}):
$k_1 \mapsto h_0,\quad  x^+_1  \mapsto e_0, \quad x^-_1  \mapsto f_0,\quad
k_0  \mapsto -c-h_0, \quad x^+_0   \mapsto f_{-1}, \quad x^-_0  \mapsto e_{1},\quad
c\mapsto - c.$\vspace{1mm}

We give now an equivalent presentation of $\widehat{sl_2}$  using the generic construction exposed in Section \ref{sec:FRT}.
In this goal, let us introduce the following classical  (traceless) r-matrix for $x\neq 1$ associated with the affine Kac-Moody algebra $\widehat{sl_2}$:
\begin{equation}\label{def:r}
 r(x)=\frac{1}{x-1}\begin{pmatrix}
       -\frac{1}{2}(x+1)&0&0&0\\
       0&\frac{1}{2}(x+1)& -2&0\\
       0&-2x& \frac{1}{2}(x+1) &0\\
        0&0&0&-\frac{1}{2}(x+1)
      \end{pmatrix}
\end{equation}
which satisfies the classical Yang-Baxter equation \eqref{eq:CYBE} and the following relations
\begin{equation}
 r_{12}(x)=-r_{21}(1/x)=-r_{12}(1/x)^{t_1t_2}\;.
\end{equation}
%
%

The affine Kac-Moody algebra $\widehat{sl_2}$ admits a third presentation, called an FRT presentation, using the results of the previous section\footnote{We expect this presentation appears in the literature, although we could not find a reference.}. Namely, by defining:
\begin{eqnarray}
 T^{+}(x)&=&   \begin{pmatrix}
          h_0/2 &2 f_0 \\
       0 &   -h_0/2
                   \end{pmatrix}   + \sum_{n\geq 1}x^n\begin{pmatrix}
           h_{n} &2f_{n} \\
        2 e_{n}&   -h_{n}
                   \end{pmatrix},\label{Tp}
                   \\
 T^{-}(x)&=& \begin{pmatrix}
         - h_0/2 & 0 \\
       -2e_0 &   h_0/2
                   \end{pmatrix}   + \sum_{n\geq 1}x^{-n} \begin{pmatrix}
       -  h_{- n} & -2f_{-n}\\
        -2 e_{-n} &   h_{- n} 
                   \end{pmatrix},\label{Tm}
\end{eqnarray}
the relations given in Definition \ref{def:T} are equivalent to the relations \eqref{CW1}-\eqref{CW3}. 
Let us remark that the traces of $T^\pm(x)$ vanish.

\subsection{Automorphisms of $\widehat{sl_2}$ and fixed point subalgebras}

We now describe explicitly the automorphisms obtained from Proposition \ref{prop:tw} in the special case of $\widehat{sl_2}$. 
\begin{prop}\label{pr:U} Let the r-matrix be defined by (\ref{def:r}). The only solutions of the equations (\ref{eqU}) are given by:
\begin{equation}
\label{solUg}
U=\begin{pmatrix}
       k&0\\
       0&-k^*
      \end{pmatrix}
 \quad \mbox{for $\epsilon=1$} \qquad \mbox{or}\qquad U(x)=\begin{pmatrix}
      0 &1/\sqrt{x}\\
       \pm\sqrt{x}&0
      \end{pmatrix} 
\;   \quad \mbox{for $\epsilon=\pm 1$},
\end{equation}
where $k,k^*$ are non-zero scalar parameters.
\end{prop}
\begin{proof} By direct computation.
\end{proof}

In the following, we study the automorphisms of $\widehat{sl_2}$ denoted by $\theta_1$ (resp. $\theta_2$) 
obtained from \eqref{eq:auU} with the particular solution for (\ref{eqU})
\begin{equation}\label{solU}
U=\begin{pmatrix}
       1&0\\
       0&-1
      \end{pmatrix}
 \quad \text{(resp. } U(x)=\begin{pmatrix}
      0 &1/\sqrt{x}\\
       -\sqrt{x}&0
      \end{pmatrix}).
\end{equation}
We restrict ourselves to these particular solutions since the subalgebras obtained 
for the more general cases of $U(x)$ \eqref{solUg} are isomorphic to the ones we get from \eqref{solU}. 
\vspace{1mm}

 The following corollary is deduced from
 the explicit form of $T^\pm(x)$ given in \eqref{Tp}-\eqref{Tm} and relation \eqref{eq:auU}:
\begin{cor}\label{cor:twloop} The action of the automorphisms $\theta_1,\theta_2$ on the generators 
in the Cartan-Weyl basis of $\widehat{sl_2}$ is such that:
\begin{alignat}{10}
&\theta_1(e_n)=f_{-n},\qquad &&\theta_1(f_n)=e_{-n},\qquad&& \theta_1(h_n)=-h_{-n}\quad&& \text{and}\qquad &&\theta_1(c)=-c\;,\label{eq:theta1}\\
& \theta_2(e_n)= e_{-n+1},\qquad && \theta_2(f_n)= f_{-n-1},\qquad &&\theta_2(h_n)=h_{-n} + c\delta_{n}\quad && \text{and}\qquad&&\theta_2(c)=-c\;.\label{eq:theta2}
\end{alignat}
\end{cor}
In the Serre-Chevalley basis, these automorphisms are given by $\theta_1(x^\pm_i)=x^\mp_i$, $\theta_1(k_i)=-k_{i}$
and
$\theta_2(x^\pm_i)= x^\mp_{1-i}$, $\theta_2(k_i)=-k_{1-i}$.

Note that solely solving the relations \eqref{eqU}, two well-known automorphisms of $\widehat{sl_2}$ are recovered. 
In the literature, the automorphism $\theta_1$ is called the  Chevalley involution, whereas the automorphism $\theta_2$ is a composition of the outer automorphism of $\widehat{sl_2}$ and the Chevalley involution.
\vspace{2mm}

According to the choice of automorphisms of $\widehat{sl_2}$, an FRT presentation for two different fixed point subalgebras of the affine Kac-Moody algebra 
$\widehat{sl_2}$ can be constructed. First, we consider the fixed point subalgebra $\widehat{sl_2}^{\theta_1}$.  
Using (\ref{Tp}), (\ref{Tm}), by straightforward simplifications the element (\ref{eq:Bxg}) is written as\footnote{The central charge $c$ does not appear, as the derivative of $U$ is vanishing.}:
\begin{equation}\label{eq:B1} 
 B(x)=\begin{pmatrix}
          0& \overline{A}_0\\
          0&0
           \end{pmatrix}+\sum_{n\geq 1}x^n\begin{pmatrix}
          \overline{G}_{n}& \overline{A}_{-n}\\
          \overline{A}_{n}&-\overline{G}_{n} \end{pmatrix}\;
 \mapsto T^+(x) + \theta_1(T^+(x))=T^+(x) + U\ T^-(1/x)^t\ U^{-1},
\end{equation} 
where the generators $\overline{A}_{n},\overline{G}_{n}$ are such that:
\begin{equation}\label{eq:OAdav}
 \overline{A}_n \mapsto 2(  e_n +  f_{-n})\quad\text{and}\quad \overline G_n \mapsto h_n-h_{-n}\;.
\end{equation}
\vspace{1mm}
 
Now, consider the second solution of (\ref{solU}). For this choice of automorphism, the element $B(x)$ is written as:
\begin{eqnarray}\label{eq:B2}
 B(x) &=&\begin{pmatrix}
          \overline{K}_0/2&\overline{Z}_0^- \\
          0&-\overline{K}_0/2
           \end{pmatrix}+\sum_{n\geq 1}x^n\begin{pmatrix}
          \overline{K}_{n}& \overline{Z}_{n}^-\\
          \overline{Z}_{n}^+&-\overline{K}_{n} \end{pmatrix}\; \\
&\mapsto& T^+(x) + \theta_2(T^+(x))=T^+(x) + U(x)\ T^-(1/x)^t\ U^{-1}(x)  -c x U'(x) U(x)^{-1},\nonumber
\end{eqnarray}
where the generators $\overline{K}_{n},\overline{Z}_{n+1}^+, \overline{Z}_{n}^-$, for $n\in \ZZ_{\geq 0}$, are such that:
\begin{equation}
 \overline{Z}^+_{n+1}\mapsto 2( e_{n+1} + e_{-n}),\quad \overline{Z}^-_n\mapsto 2(f_n + f_{-n-1}),\quad \mbox{and}\quad \overline{K}_{n}\mapsto h_{n} +h_{-n}+ c\delta_{n}\;.\label{eq:isOAaug}
\end{equation}
\vspace{1mm}

\subsection{Classical reflection equation for $\widehat{sl_2}$ and fixed point subalgebras}

Previously, we identified two different subalgebras of $\cT$  as fixed point subalgebras of $\widehat{sl_2}$ under the action of certain automorphisms. 
In this subsection,   using Proposition \ref{propB2} we  obtain two other known fixed point subalgebras of $\widehat{sl_2}$ .\vspace{1mm}

\begin{prop} Let the r-matrix be defined by (\ref{def:r}). The most general k-matrix solution of the classical reflection equation (\ref{eq:re}) is given by:
\begin{equation}\label{eq:k2}
 k(x)=\eta(x)\begin{pmatrix}
       \alpha(x-1/x)&\beta+\gamma/x\\
       -\beta-\gamma x&\delta(x-1/x)
      \end{pmatrix}\;,
\end{equation}
where $\alpha, \beta, \gamma$ and $\delta$ are scalar parameters and $\eta(x)$ is an arbitrary function.
\end{prop}
\begin{proof} By direct computation.
\end{proof}
\begin{rem} For $\alpha=1$, $\delta=-1$ and $\beta=\gamma=0$, 
(resp. $\beta= \gamma$ and $\alpha=\delta=0$)
the k-matrix \eqref{eq:k2} becomes proportional to the matrix $U$ (resp. $U(x)$) given in \eqref{solU}.
\end{rem}

By using Proposition \ref{propB2} with the k-matrix (\ref{eq:k2}), a subalgebra of $\widehat{sl_2}$ is defined depending on the 
four parameters $\alpha, \beta, \gamma$ and $\delta$. As examples, we consider the two following choices for $k(x)$.

\begin{example} \label{ex:B1}
For $k(x)\Big|_{\alpha=\delta=\gamma=0\atop \beta=\eta(x)=1}=\begin{pmatrix}
           0&1\\
           -1&0
          \end{pmatrix}=\kappa_+
$, 
the homomorphism \eqref{eq:Bxgg} becomes
\begin{equation}\label{eq:B3}
 B(x) =\begin{pmatrix}
          \overline{H}_0/2&\overline{F}_0/2 \\
          \overline{E}_0/2&-\overline{H}_0/2
           \end{pmatrix}+\sum_{n\geq 1}x^n\begin{pmatrix}
          \overline{H}_{n}& \overline{F}_{n}\\
          \overline{E}_{n}&-\overline{H}_{n} \end{pmatrix}\; \mapsto T^+(x) + \kappa_+ \ T^-(1/x)^t\ \kappa_+^{-1},
\end{equation}
where the generators $\overline{H}_{n},\overline{E}_{n}, \overline{F}_{n}$, for $n\in \ZZ_{\geq 0}$, are such that:
\begin{equation}\label{eq:Eb}
 \overline{E}_{n} \mapsto 2(e_{n} +e_{-n})\quad,\quad \overline{F}_n \mapsto 2(f_n + f_{-n})\quad\text{and}\quad  \overline{H}_{n} \mapsto h_{n} +h_{-n}\;.
\end{equation}
This subalgebra is denoted by $\widehat{sl_2}^{\kappa_+}$.
\end{example}

\begin{example} \label{ex:B2}
For $k(x)\Big|_{\alpha=\delta=\beta=0 \atop \gamma=\eta(x)=1}=\begin{pmatrix}
           0&1/x\\
           -x&0
          \end{pmatrix}=\kappa_-(x)$, the homomorphism \eqref{eq:Bxgg} becomes
\ben\label{eq:B4}
 B(x) &=& x^{-1} \begin{pmatrix}
          0&\widetilde{F}_0/2 \\
         0&0
           \end{pmatrix}+\begin{pmatrix}
          \widetilde{H}_0/2&\widetilde{F}_1 \\
          0&-\widetilde{H}_0/2
           \end{pmatrix}+x \begin{pmatrix}
          \widetilde{H}_1&\widetilde{F}_2 \\
          \widetilde{E}_0/2&-\widetilde{H}_1
           \end{pmatrix}+\sum_{n\geq 2}x^n\begin{pmatrix}
          \widetilde{H}_{n}& \widetilde{F}_{n+1}\\
          \widetilde{E}_{n-1}&-\widetilde{H}_{n} \end{pmatrix}\nonumber\; \\
&\mapsto&  T^+(x) + \kappa_-(x)\ T^-(1/x)^t\ \kappa_-(x)^{-1}-c x \kappa_-'(x) \kappa_-(x)^{-1} ,\nonumber
\een
where the generators $\widetilde{H}_{n},\widetilde{E}_{n}, \widetilde{F}_{n}$, for $n\in \ZZ_{\geq 0}$, are such that:
\begin{equation}\label{eq:Et}
\widetilde{E}_{n}\mapsto 2( e_{n+1} + e_{-n+1})\quad,\quad \widetilde{F}_n\mapsto 2(f_{n-1} +f_{-n-1})\quad\text{and}\quad  \widetilde{H}_{n}\mapsto h_{n} +h_{-n}+2c \delta_{n}\;.
\end{equation}
This subalgebra is denoted by $\widehat{sl_2}^{\kappa-}$.
\end{example}
By using the commutation relations \eqref{CW1}-\eqref{CW3}, one finds that both $\widehat{sl_2}^{\kappa_\pm}$  are isomorphic Lie algebras. 
Note also that $\widehat{sl_2}^{\kappa_+}$ (resp.  $\widehat{sl_2}^{\kappa_-}$) can be viewed as the fixed point subalgebra of $\widehat{sl_2}$ under the action of the Lusztig automorphism
in the Cartan-Weyl presentation
$e_n\mapsto e_{-n}$, $f_n\mapsto f_{-n}$,  $h_n\mapsto h_{-n}$ and $c\mapsto -c$ 
(resp. $e_n\mapsto e_{-n+2}$, $f_n\mapsto f_{-n-2}$,  $h_n\mapsto h_{-n}+2c\delta_{n}$ and $c\mapsto -c$). However,
let us emphasize that the Lusztig automorphisms cannot be written in the form of
Proposition \ref{prop:tw}.

\vspace{1mm}

\section{The Onsager algebras and current algebras}
 In this section, it is shown that the three non-standard classical Yang-Baxter algebras (\ref{eq:Al}) associated with 
(\ref{eq:B1}), (\ref{eq:B2}) and   (\ref{eq:B3})  provide an FRT presentation for the Onsager algebra, augmented Onsager algebra and the $sl_2$-invariant Onsager algebra, respectively. 
In each case, the corresponding current presentations are derived. Using the FRT presentation, we also derive the corresponding commutative subalgebras. In particular, this provides a new derivation of the well-known mutually commuting quantities in integrable models generated from the Onsager algebra such as the Ising \cite{Ons44} or superintegrable chiral Potts models \cite{RG,Davies}. For the augmented Onsager algebra, it gives classical analogs of the mutually commuting quantities constructed in \cite{BB3}.

\subsection{The Onsager algebra revisited\label{sec:Ons}}
Introduced in the context of mathematical physics on the exact solution of the two-dimensional Ising model \cite{Ons44}, the Onsager algebra is known to admit at least two presentations. The first presentation which originates in Onsager's work \cite{Ons44} is now recalled. 
\begin{defn}\label{def:OA}
 The Onsager algebra $\cO$ is generated by $\{A_n,G_m|n,m \in\ZZ\}$ subject to the following relations:
 \begin{eqnarray}
 &&[A_n,A_m]=4\ G_{n-m}\;,\label{eq:OA1}\\
 &&[G_n,A_m]=2A_{n+m}-2A_{m-n}\;,\\
  &&[G_n,G_m]=0\label{eq:OA3}\;.
 \end{eqnarray}
\end{defn} 
As the Lie bracket is anti-symmetric, $[A_0,A_n]+[A_n,A_0]=0$. From (\ref{eq:OA1}), note that
\beqa
G_{-n}+G_n=0.\label{eq:Glin}
\eeqa

Note that a second presentation is given in terms of two generators $A_0,A_1$ subject to a pair of relations, the so-called
Dolan-Grady relations \cite{DG82}. They read:
 \begin{eqnarray}
[A_0,[A_0,[A_0,A_1]]]=16[A_0,A_1], \qquad [A_1,[A_1,[A_1,A_0]]]=16[A_1,A_0].\label{eq:DG}
 \end{eqnarray}

\begin{thm}
The non-standard classical Yang-Baxter algebra (\ref{eq:Al}) specialized for
\begin{equation}
 B(x)=\begin{pmatrix}
       {\cal G}(x) &{\cal A}^-(x)\\
       {\cal A}^+(x) & -{\cal G}(x)
      \end{pmatrix}\label{eq:BO}
\end{equation}
with 
\begin{eqnarray}
{\cal G}(x)=\sum_{n\geq 1} x^n G_{n},\quad \ {\cal A}^-(x)=\sum_{n\geq 0} x^n A_{-n}
 ,\quad \ {\cal A}^+(x)=\sum_{n\geq 1} x^n A_{n}\;\label{eq:cu}
\end{eqnarray}
and the r-matrix given by
 \begin{equation}
   \overline{r}_{12}(x,y)=r_{12}(x/y)+U_1\ r_{12}^{t_1}(1/(xy))\ U_1^{-1} \quad \text{where $U$ is given by \eqref{solU}}, \label{rOnsager}
 \end{equation}
provides an FRT presentation of the Onsager algebra.
\end{thm}

\begin{proof} Inserting (\ref{eq:BO}) into (\ref{eq:Al}), one has:
\begin{eqnarray}
 &&[\ {\cal G}(x)\ , \ {\cal G}(y)\ ]=0\;,\label{eq:GG}\\
 &&[\ {\cal G}(x)\ , \ {\cal A}^+(y)\ ]=\frac{2x(1-y^2)}{(x-y)(xy-1)}{\cal A}^+(y)+\frac{2y}{x-y}{\cal A}^+(x)
 +\frac{2 xy}{xy-1}{\cal A}^{-}(x)\;,\label{eq:GAp}\\
 &&[\ {\cal G}(x)\ , \ {\cal A}^-(y)\ ]=\frac{2x(y^2-1)}{(x-y)(xy-1)}{\cal A}^-(y)+\frac{2x}{y-x}{\cal A}^-(x)
 +\frac{2}{1-xy}{\cal A}^{+}(x)\;,\label{eq:GAm}\\
 &&[\ {\cal A}^+(x)\ , \ {\cal A}^{+}(y)\ ]=\frac{4xy}{1-xy}({\cal G}(x)-{\cal G}(y))\quad,\quad [\ {\cal A}^+(x)\ , \ {\cal A}^{-}(y)\ ]=\frac{4x}{x-y}({\cal G}(x)-{\cal G}(y))\;,\label{eq:AA}\\
 &&[\ {\cal A}^-(x)\ , \ {\cal A}^{-}(y)\ ]=\frac{4}{xy-1}({\cal G}(x)-{\cal G}(y))\;\label{eq:AAmp}.
 \end{eqnarray}
 Define
 \begin{eqnarray}
{\cal G}(x)=\sum_{n\geq 1} x^n \overline{G}_{n},\quad \ {\cal A}^-(x)=\sum_{n\geq 0} x^n \overline{A}_{-n}
 ,\quad \ {\cal A}^+(x)=\sum_{n\geq 1} x^n \overline{A}_{n}\;.\label{eq:cu}
\end{eqnarray}
Then, we extract from (\ref{eq:GG})-(\ref{eq:AAmp}) the complete set of relations satisfied by the generators $ \overline{A}_{n}, \overline{G}_{m}$ for $n,m\geq 1$. Consider (\ref{eq:AA}),   (\ref{eq:AAmp}). Expanding around $y=0$ both sides of the three equations and identifying the power series, one obtains equivalently:
\beqa
\big[\oA_n,\oA_m\big] &=& 4sign(n-m)\oG_{|n-m|} ,\qquad \big[\oA_{-n},\oA_{-m}\big] = 4sign(m-n)\oG_{|n-m|} ,\nonumber\\
\big[\oA_{n},\oA_{-m}\big] &=& 4\oG_{n+m}, 
\qquad \qquad \qquad  \qquad   \big[\oA_n,\oA_0\big]=4\oG_n\quad \qquad  \qquad  \qquad \mbox{for any} \quad n,m\geq 1.\nonumber
\eeqa
Consider the r.h.s of (\ref{eq:GAp}). Around $y=0$, one has:
\beqa
\frac{2x(1-y^2)}{(x-y)(xy-1)}{\cal A}^+(y)= -2\sum_{m=1}^\infty\sum_{n=0}^{m-1}x^{-n}y^m\oA_{m-n} - 2\sum_{m=1}^\infty\sum_{n=1}^mx^{n}y^m\oA_{m-n}\nonumber 
\eeqa
whereas 
\beqa
\frac{2y}{x-y}{\cal A}^+(x)+\frac{2 xy}{xy-1}{\cal A}^{-}(x) &=& \sum_{n=1}^{\infty}\sum_{m=1}^{\infty}x^ny^m (2\oA_{n+m}-2\oA_{m-n}) \nonumber\\
&&+ \ 2\sum_{m=1}^\infty\sum_{n=0}^{m-1}x^{-n}y^m\oA_{m-n} + 2\sum_{m=1}^\infty\sum_{n=1}^mx^{n}y^m\oA_{m-n}.\nonumber
\eeqa
Combining all terms together with the l.h.s of (\ref{eq:GAp}), one gets equivalently:
\beqa
\big[\oG_n,\oA_m\big] = 2\oA_{m+n} -2\oA_{m-n} \quad \mbox{for any} \quad n,m\geq 1\nonumber.
\eeqa
The similar analysis for (\ref{eq:GAm}) gives:
\beqa
\big[\oG_n,\oA_{-m}\big] = 2\oA_{-m+n} -2\oA_{-m-n}\quad \mbox{and}\quad \big[\oG_n,\oA_{0}\big] = 2\oA_{n} -2\oA_{-n}\quad \mbox{for any} \quad n,m\geq 1\nonumber.
\eeqa
Finally, from (\ref{eq:GG}) we immediatly obtain:
\beqa
\big[\oG_n,\oG_{m}\big] =0\quad \mbox{for any} \quad n,m\geq 1\nonumber.
\eeqa
It remains to show that the algebra generated by $\oA_n,\oG_m$ is isomorphic to the Onsager algebra.
With the identification  
\beqa
A_n = \oA_n \quad \mbox{and}\quad G_n= sign(n)\oG_{|n|}\quad \mbox{for any $n,m\in{\mathbb Z}$}, 
\eeqa
one obtains the defining relations of the Onsager algebra (\ref{eq:OA1})-(\ref{eq:OA3}).
\end{proof}

\begin{rem}  
The fixed point subalgebra of the loop algebra of  $sl_2$ under the action of $\theta_1$  (see relations \eqref{eq:theta1}) is isomorphic to the Onsager algebra $\cO$  \cite{Davies,Roan}. For $\widehat{sl_2}^{\theta_1}$, the isomorphism is given by (\ref{eq:OAdav}).
\end{rem}
Note that the relations (\ref{eq:GG})-(\ref{eq:AAmp}) provide a current presentation for the Onsager algebra.\vspace{2mm}
 
Using the general results presented in Section \ref{sec:CS} (see Propositions \ref{pr:a1}, \ref{pr:a2}), 
the generating functions of the elements in the 
commutative subalgebras of the Onsager algebra are easily derived. On one hand, by Proposition \ref{pr:a1} we routinely obtain:
\begin{prop} The generating function 
\begin{equation} 
 t^{ons}(x)=2{\cal G}(x)^2+{\cal A}^+(x){\cal A}^-(x)+{\cal A}^-(x){\cal A}^+(x)\;
\end{equation}
is such that $[t^{ons}(x),t^{ons}(y)]=0$ for any $x,y$.
\end{prop}

On the other hand, by Proposition \ref{pr:a2} it follows:
\begin{prop} Let $\kappa,\kappa^*,\mu$ be arbitrary scalars. The generating function 
\begin{equation} \label{eq:t1}
 b^{ons}(x)=(\kappa+\kappa^*/x){\cal A}^+(x)+(\kappa + \kappa^*x){\cal A}^-(x)+\mu(1/x-x){\cal G}(x)
\end{equation}
is such that $[b^{ons}(x),b^{ons}(y)]=0$ for any $x,y$.
\end{prop}
\begin{proof} First, one shows that
\begin{equation}
  M(x)=\begin{pmatrix}
        \mu/x&\kappa+\kappa^*/x\\
        \kappa + \kappa^* x &\mu x
       \end{pmatrix}\;
 \end{equation}
is a solution of the relation \eqref{eq:reD} with the r-matrix given by \eqref{rOnsager}. From (\ref{tb2}), we immediately get (\ref{eq:t1}).
\end{proof}
\begin{rem} Expanding the generating function $b^{ons}(x)$ in $x$, one produces the well-known mutually commuting quantities of the Onsager algebra $\cO$ \cite{Ons44,RG,Davies}. For $k=0,1,2,\dots$, the coefficients of the power series are given by:
\begin{equation}
 I_k=\kappa(A_k+A_{-k}) +\kappa^* (A_{k+1}+A_{-k+1})+\mu( G_{k+1}-G_{k-1})\;.
\end{equation}
\end{rem}
\vspace{1mm}

\subsection{The augmented Onsager algebra revisited \label{sec:aOns}}
The augmented Onsager algebra has been introduced in \cite{BC13} as a classical analog of the  augmented tridiagonal algebra of the first kind \cite[page 5-6]{IT}, also called the augmented $q-$Onsager algebra in \cite{BB3,BB17}.
\begin{defn}\label{def:augOA} \cite{BC13}
 The augmented Onsager algebra $\cO^{aug}$ is generated by $\{K_n,Z^\pm_{m}|n,m \in\ZZ\}$  subject to the following relations
 \begin{eqnarray}
 &&[Z^+_n,Z^+_m]=[Z^{-}_n,Z^{-}_m]=[K_n,K_m]=0\ ,\label{eq:OAaug3}\\
 && [Z^+_n,Z^{-}_m]=4(K_{n+m}+ K_{-n+m+1})\ ,\qquad [K_n,Z^\pm_m]=\pm 2(Z^\pm_{n+m}+Z^\pm_{-n+m})\ ,\label{eq:OAaug1}\\
&& K_n-K_{-n}=0\ ,\quad Z_n^+-Z_{-n+1}^+=0 \quad \mbox{and}\quad  Z_n^--Z_{-n-1}^- =0\ .\label{eq:OAaug2}
 \end{eqnarray}
\end{defn} 
Let us point out that, from the defining relations (\ref{eq:OAaug3})-(\ref{eq:OAaug1}), it is possible to show that the  three generators $K_0,Z^\pm_0$ satisfy:
 \begin{eqnarray}
[Z^\pm_0,[Z^\pm_0,[Z^\pm_0,Z^\mp_0]]]=0, \qquad [K_0,Z^\pm_0]=\pm 4 Z^\pm_0\label{eq:DGaug}.
 \end{eqnarray}
These latter relations can be viewed as the classical analog of the relations in \cite[page 5]{IT} or \cite[eqs. (3.22)]{BB3}.\vspace{1mm}

\begin{thm}
The non-standard classical Yang-Baxter algebra (\ref{eq:Al}) specialized for
\begin{equation}
 B(x)=\begin{pmatrix}
       {\cal K}(x) &{\cal Z}^-(x)\\
       {\cal Z}^+(x) & -{\cal K}(x)
      \end{pmatrix}\label{eq:BOaug}
\end{equation}
with 
\begin{equation}
{\cal K}(x)=K_{0}/2+\sum_{n\geq 1} x^n K_{n}\quad,\quad {\cal Z}^-(x)=\sum_{n\geq 0} x^n Z^-_{n}
 \quad,\quad {\cal Z}^+(x)=\sum_{n\geq 1} x^n Z^+_{n}\;, \label{eq:caug}
\end{equation}
and the r-matrix given by
 \begin{equation}
   \overline{r}_{12}(x,y)=r_{12}(x/y)+U_1(x)\ r_{12}^{t_1}(1/(xy))\ U_1(x)^{-1} \quad \text{where $U(x)$ is given by \eqref{solU}}, \label{raugOnsager}
 \end{equation}
provides an FRT presentation of the augmented Onsager algebra.
\end{thm}
\begin{proof} 
Inserting (\ref{eq:BOaug}) into (\ref{eq:Al}), one has:
\begin{eqnarray}
&&[\ {\cal K}(x)\ ,\ {\cal K}(y)\ ]=0\quad , \quad [\ {\cal Z}^+(x)\ ,\ {\cal Z}^+(y)\ ]=0\quad , \quad [\ {\cal Z}^-(x)\ ,\ {\cal Z}^-(y)\ ]=0\quad ,\label{eq:KK}\\
&&[\ {\cal K}(x)\ ,\ {\cal Z}^+(y)\ ]=\frac{2y(x+1)(y-1)}{(x-y)(xy-1)}{\cal Z}^+(x)-\frac{2y(x^2-1)}{(x-y)(xy-1)}{\cal Z}^+(y)\quad ,\label{eq:KZp} \\
&&[\ {\cal K}(x)\ ,\ {\cal Z}^-(y)\ ]=\frac{2y(x^2-1)}{(x-y)(xy-1)}{\cal Z}^-(y)-\frac{2x(x+1)(y-1)}{(x-y)(xy-1)}{\cal Z}^-(x)\quad ,\label{eq:KZm} \\
&&[\ {\cal Z}^+(x)\ ,\ {\cal Z}^-(y)\ ]=\frac{4x}{(x-y)(xy-1)} \big( (x+1)(y-1) {\cal K}(x)\ -\ (x-1)(y+1){\cal K}(y) \big)\;.\label{eq:ZZ}
\end{eqnarray}
Define
\begin{equation}
{\cal K}(x)=\oK_{0}/2+\sum_{n\geq 1} x^n \oK_{n}\quad,\quad {\cal Z}^-(x)=\sum_{n\geq 0} x^n \oZ^-_{n}
 \quad,\quad {\cal Z}^+(x)=\sum_{n\geq 1} x^n \oZ^+_{n}\;. \label{eq:caug}
\end{equation}
Then, we extract from (\ref{eq:KK})-(\ref{eq:ZZ}) the complete set of relations satisfied by the generators $ \oK_{n}, \oZ^-_{n},\oZ^+_{m}$ for $n\geq 0,m\geq 1$. Consider (\ref{eq:ZZ}).  Expanding around $y=0$, equivalently we get:
\beqa
\big[\oZ^+_n,\oZ^-_m\big]&=& 4(\oK_{n+m} + \oK_{|n-m-1|}) \quad \mbox{for any} \quad n, m \geq 1.\nonumber
\eeqa
Consider (\ref{eq:KZp}). Around $y=0$, we get:
\beqa
\big[\oK_n,\oZ_m^+\big]&=&2(\oZ_{n+m}^+ + \oZ_{n-m+1}^+) \quad \mbox{for any}\quad 1\leq m\leq n,\nonumber\\
\big[\oK_n,\oZ_m^+\big]&=&2(\oZ_{n+m}^+ + \oZ_{-n+m}^+) \quad\  \mbox{for any}\quad m> n \geq 0.\nonumber
\eeqa
Consider  (\ref{eq:KZm}). Around $y=0$, we get:
\beqa
\big[\oK_n,\oZ_m^-\big]&=&- 2(\oZ_{n+m}^- + \oZ_{-n+m}^-) \quad \ \ \mbox{for any}\quad m\geq n \geq 0,\nonumber\\
\big[\oK_n,\oZ_m^-\big]&=&- 2(\oZ_{n+m}^- + \oZ_{n-m-1}^-) \quad\  \mbox{for any}\quad 1\leq m< n,\nonumber\\
\big[\oK_n,\oZ_0^-\big]&=&-2(\oZ_n^-+\oZ_{n-1}^- )\quad \mbox{for any}\quad n\geq 1.\nonumber
\eeqa
Finally, from (\ref{eq:KK}) we immediately obtain:
\beqa
\big[\oK_n,\oK_{n'}\big] =0,\quad \big[\oZ^+_m,\oZ^+_{m'}\big] =0, \quad \big[\oZ^-_n,\oZ^-_{n'}\big] =0  \quad \mbox{for any} \quad n,n'\geq 0,\ m,m'\geq 1\nonumber.
\eeqa
It remains to show that the algebra generated by $\oK_n,\oZ_m^\pm$ is isomorphic to the augmented Onsager algebra.
With the identification  
\beqa
 K_n = \oK_{|n|}, \quad Z_n^+&=&\oZ_{n}^+ \quad \mbox{if}\quad n\geq 1,\quad Z_n^+=\oZ_{-n+1}^+ \quad \mbox{if}\quad n\leq 0 ,\nonumber\\
 Z_n^-&=&\oZ_{n}^- \quad \mbox{if}\quad n\geq 0, \quad  Z_n^-=\oZ_{-n-1}^- \quad \mbox{if}\quad n< 0,\nonumber
\eeqa
one obtains the defining relations of the augmented Onsager algebra (\ref{eq:OAaug3})-(\ref{eq:OAaug1}).
\end{proof}
\begin{cor}    The augmented Onsager algebra is isomorphic to the  fixed point subalgebra $\widehat{sl_2}^{\theta_2}$ generated by (\ref{eq:isOAaug}).
\end{cor}
Note that the relations (\ref{eq:KK})-(\ref{eq:ZZ}) provide a current presentation for the augmented Onsager algebra.
\vspace{2mm}

Generating functions of elements in the commutative subalgebras of the augmented Onsager algebra are now derived. By  Propositions \ref{pr:a1} we routinely obtain:
\begin{prop} The generating function 
\begin{equation} 
 t^{aug}(x)=2{\cal K}(x)^2+{\cal Z}^+(x){\cal Z}^-(x)+{\cal Z}^-(x){\cal Z}^+(x)\;
\end{equation}
is such that $[t^{aug}(x),t^{aug}(y)]=0$ for any $x,y$.
\end{prop}

On the other hand, by Proposition \ref{pr:a2} it follows:
\begin{prop} Let $\tau,\nu,\nu^*$ be arbitrary scalars. The generating function 
\begin{equation} \label{eq:t2}
 b^{aug}(x)= \tau{\cal K}(x)+\nu(1+1/x){\cal Z}^+(x)+\nu^*(x+1){\cal Z}^{-}(x)
\end{equation}
is such that $[b^{aug}(x),b^{aug}(y)]=0$ for any $x,y$. 
\end{prop}
\begin{proof} One shows that
\begin{equation}
 M(x)=\begin{pmatrix}
        \tau&\nu(1+1/x)\\
       \nu^*(x+1)&0
       \end{pmatrix}
 \end{equation}
is a solution of the relation \eqref{eq:reD} with the r-matrix \eqref{raugOnsager}. From (\ref{tb2}), we immediately get (\ref{eq:t2}).
\end{proof}
\begin{rem} Expanding the generating function $b^{aug}(x)$ in $x$, one produces mutually commuting quantities of the augmented Onsager algebra $\cO^{aug}$. 
For $k=0,1,2,\dots$ and the convention $Z_0^+=Z_{-1}^-=0$, the coefficients of the power series read:
\begin{equation}
 I^{aug}_k=\tau K_k +\nu( Z_k^+ + Z^+_{k+1}) +\nu^*( Z_{k-1}^- + Z^-_{k}).
\end{equation}
\end{rem}
\vspace{1mm}

\subsection{The $sl_2$-invariant Onsager algebra\label{sec:inOns}} In this subsection, we introduce an algebra that we call the $sl_2$-invariant algebra. This algebra is viewed as the classical analog of the $U_q(gl_2)$-invariant $q-$Onsager algebra\footnote{In the context of the half-infinite XXZ spin chain, the $U_q(gl_2)$-invariant $q-$Onsager algebra algebra characterizes the hidden non-Abelian infinite dimensional symmetry of the
 Hamiltonian with $U_q(gl_2)$-invariant (special diagonal) boundary conditions (for the finite open XXZ spin chain, see \cite{PS}). This motivates the terminology used in \cite{BB17} and here.
In the defining relations of the $sl_2$-invariant algebra, one recognizes the $sl_2$-subalgebra (\ref{eq:sl2O0}).} introduced in \cite[subsection 2.4]{BB17}.
\begin{defn}\label{def:invOA}
The $sl_2$-invariant Onsager algebra $\cO^{inv}$ is generated by $\{H_n,E_{n},F_n|n \in\ZZ\}$  subject to the following relations
 \begin{eqnarray}
 && [E_n,E_m]=[F_n,F_m]=[H_n,H_m]=0,\label{eq:OIsl1}\\
 && [H_n,E_m]=2(E_{n+m}+E_{-n+m})\ ,\quad [H_n,F_m]=-2(F_{n+m}+F_{-n+m})\;,\label{eq:OIsl11}\\
 && [E_n,F_m]=H_{n+m}+H_{-n+m}\ .\label{eq:OIsl21}
 \end{eqnarray}
\end{defn} 

As a consequence of the commutation relations (\ref{eq:OIsl1})-(\ref{eq:OIsl21}), linear relations among the generators $\{E_n,F_n,H_n|{n\in \mathbb Z}\}$ can be exhibited. Define:
\ben
\E_{n}=E_n-E_{-n}, \quad \F_{n}=F_n-F_{-n}, \quad \HH_{n}=H_n-H_{-n}.
\een
From (\ref{eq:OIsl21}) and (\ref{eq:OIsl1})  one gets
\ben\label{cEF1}
[F_n,\E_m]=[E_n,\F_m]=[E_n,\E_m]=[F_n,\F_m]=0.\label{eq:c0}
\een
Using (\ref{eq:OIsl21}) with $n=0$ we have $H_m=\frac{1}{2}[E_0,F_m]$. Applying the Jacobi identify we obtain
\ben
[H_n,\E_m]=\frac{1}{2}[[E_0,F_m] ,\E_m]=0, \qquad
[H_n,\F_m]=\frac{1}{2}[[E_0,F_m] ,\F_m]=0. \label{eq:c1}
\een
Also, using the commutation relations (\ref{eq:OIsl11}),  we find
\ben\label{Re1}
[\HH_n,F_m]=[\HH_n,E_m]=[\HH_n,H_m]=0.\label{eq:c2}
\een
Thus, according to (\ref{eq:c0}), (\ref{eq:c1}) and (\ref{eq:c2}) the elements $\{\HH_{n},\E_{n},\F_{n}\}$ with $n\in \ZZ$ belong to the center of $\cO^{inv}$.
Finally, observe that the commutation relations  (\ref{eq:OIsl11}) imply
\ben
[H_n,\F_m]=-2(\F_{n+m}+\F_{-n+m}),\qquad [H_n,\E_m]=2(\E_{n+m}+\E_{-n+m}).
\een
Together with (\ref{eq:c1}), it follows $\F_{n+m}+\F_{-n+m}=\E_{n+m}+\E_{-n+m}=0$ for all $n$ and $m$. In particular, for $n=0$ one gets $\E_{m}=\F_{m}=0$. It implies $\HH_m=\frac{1}{2}[E_0,\F_m]=0$. As a consequence, in the algebra $\cO^{inv}$ the following linear relations hold:
\ben
E_n-E_{-n}=0, \quad F_n-F_{-n}=0, \quad H_n-H_{-n}=0.\label{eq:EFHlin}
\een

From the defining relations (\ref{eq:OIsl1})-(\ref{eq:OIsl21}), one finds that the six generators $H_0,E_0, F_0,H_1,E_1, F_1$ satisfy the relations:
 \ben
&&\null[H_0,E_0]=4E_0, \quad [H_0,F_0]=-4F_0,\quad [E_0,F_0]=2H_{0},\label{eq:sl2O0}\\
&&\null[H_0,E_1]=[H_1,E_0]=4E_1, \quad [H_0,F_1]=[H_1,F_0]=-4F_1,\quad [E_0,F_1]=[E_1,F_0]=2H_{1},\\
&&\null[H_1,[E_1,F_1]]=0.\label{eq:sl2O3}
 \een
Note that these relations are the classical analogs of \cite[eqs. (2.18)]{BB17}.

\begin{thm}
The non-standard classical Yang-Baxter algebra (\ref{eq:Al}) specialized for
\begin{equation}
 B(x)=\begin{pmatrix}
       {\cal H}(x) &{\cal F}(x)\\
       {\cal E}(x) & -{\cal H}(x)
      \end{pmatrix}\label{eq:BOinv}
\end{equation}
with 
\begin{equation}
{\cal H}(x)=H_{0}/2+\sum_{n\geq 1} x^n H_{n}\quad,\quad {\cal E}(x)=E_0/2+\sum_{n\geq 1} x^n E_{n}
 \quad,\quad {\cal F}(x)=F_0/2+\sum_{n\geq 1} x^n F_{n}\;,
\end{equation}
and the r-matrix given by
 \begin{equation}
   \overline{r}_{12}(x,y)=r_{12}(x/y)+\kappa_+ \ r_{12}^{t_1}(1/(xy))\ \kappa_+^{-1} \quad \text{where \ $\kappa_+$ \ is given in Example \ref{ex:B1}}, \label{rinvOnsager}
 \end{equation}
provides an FRT presentation of the $sl_2$-invariant Onsager algebra.
\end{thm}
\begin{proof} 
Inserting (\ref{eq:BOinv}) into (\ref{eq:Al}), one has:
\begin{eqnarray}
&&[\ {\cal H}(x)\ ,\ {\cal H}(y)\ ]=0\quad , \quad [\ {\cal E}(x)\ ,\ {\cal E}(y)\ ]=0\quad , \quad [\ {\cal F}(x)\ ,\ {\cal F}(y)\ ]=0\quad ,\label{eq:HH}\\
&&[\ {\cal H}(x)\ ,\ {\cal E}(y)\ ]=\frac{2}{(x-y)(xy-1)}\big( x(y^2-1) {\cal E}(x)-y(x^2-1){\cal E}(y)\big)\quad , \\
&&[\ {\cal H}(x)\ ,\ {\cal F}(y)\ ]=-\frac{2}{(x-y)(xy-1)}\big(x(y^2-1)  {\cal F}(x)-y(x^2-1){\cal F}(y)\big)\quad , \\
&&[\ {\cal E}(x)\ ,\ {\cal F}(y)\ ]=\frac{4}{(x-y)(xy-1)} \big( x(y^2-1)  {\cal H}(x)\ -\ y(x^2-1) {\cal H}(y) \big)\;.\label{eq:EF}
\end{eqnarray}
Define 
\begin{equation}
{\cal H}(x)=\oH_{0}/2+\sum_{n\geq 1} x^n \oH_{n}\quad,\quad {\cal E}(x)=\oE_0/2+\sum_{n\geq 1} x^n \oE_{n}
 \quad,\quad {\cal F}(x)=\oF_0/2+\sum_{n\geq 1} x^n \oF_{n}\;.
\end{equation}
As we proceed by analogy with the previous cases, we omit the details.  With the identification  
\beqa
H_n = \oH_{|n|},\quad  E_n = \oE_{|n|},\quad F_n = \oF_{|n|},  \nonumber
\eeqa
one obtains the defining relations of the $sl_2$-invariant Onsager algebra (\ref{eq:OIsl1}),  (\ref{eq:OIsl21}).
\end{proof}

\begin{cor} The $sl_2$-invariant Onsager algebra is isomorphic to the  fixed point subalgebra $\widehat{sl_2}^{\kappa_+}$ generated by \eqref{eq:Eb}. 
\end{cor}
Note that the relations (\ref{eq:HH})-(\ref{eq:EF}) provide a current presentation for the Onsager algebra.
\vspace{2mm}

Generating functions of elements in the commutative subalgebras of the $sl_2$-invariant Onsager algebra are now derived.  By  Propositions \ref{pr:a1}:
\begin{prop} The generating function 
\begin{equation} 
 t^{inv}(x)=2{\cal H}(x)^2+{\cal E}(x){\cal F}(x)+{\cal F}(x){\cal E}(x)\;
\end{equation}
is such that $[t^{inv}(x),t^{inv}(y)]=0$ for any $x,y$.
\end{prop}

By Proposition \ref{pr:a2}, it follows:
\begin{prop} Let $\mu_0, \mu_1, \mu_2$ be arbitrary scalars. The generating function 
\begin{equation} \label{eq:t3}
 b^{inv}(x)= \mu_0{\cal H}(x)+\mu_1{\cal E}(x)+\mu_2{\cal F}(x)
\end{equation}
is such that $[b^{inv}(x),b^{inv}(y)]=0$ for any $x,y$. 
\end{prop}
\begin{proof} One shows that
\begin{equation}
 M(x)=\begin{pmatrix}
        \mu_0&\mu_1\\
      \mu_2&0
       \end{pmatrix}
 \end{equation}
is a solution of the relation \eqref{eq:reD} for the r-matrix \eqref{rinvOnsager}. From (\ref{tb2}), we immediately get (\ref{eq:t3}).
\end{proof}

\vspace{2mm}

\section{Concluding remarks}

In this letter,  it is shown that the classical Onsager, augmented Onsager and $sl_2$-invariant Onsager algebras   fit into the framework of the non-standard classical Yang-Baxter algebras. Using this framework, an FRT presentation is identified in each case. A current presentation is obtained and the generating functions for elements in their commutative subalgebras are constructed.\vspace{1mm}

Here, we mainly focused on different types of non-standard classical Yang-Baxter algebras associated with one of the simplest affine Kac-Moody algebra $\widehat{sl_2}$, and their explicit relations with the Onsager algebras. However, it is clear that the framework proposed in Section 2 can be easily applied
 to any higher rank affine Lie algebra and their fixed point subalgebras. Although not discussed in this letter, an FRT presentation and associated current presentations for any of the higher rank classical Onsager algebras \cite{UI,BB} could be derived in this framework. Another generalization of the construction presented here would be to consider automorphisms of order greater than $2$. Indeed, such automorphisms have been studied for the rational r-matrix in \cite{CY07} and provide interesting integrable systems.\vspace{1mm}

In view of the recently proposed presentation of the $q-$Onsager algebra in terms of root vectors \cite{BK4} satisfying a $q-$deformed analog
of the Onsager's relations (\ref{eq:OA1})-(\ref{eq:OA3}), it is tempting to reconsider 
the connection between the quantum reflection  algebra and coideal subalgebras of $U_q(\widehat{sl_2})$ \cite{MRS} or the $q-$Onsager algebra \cite{B1,BK} 
in light of the results presented here. Besides the two Examples \ref{ex:B1}, \ref{ex:B2}, it is also natural to ask for an interpretation of the homomorphic image (\ref{eq:Bxgg}) 
in terms of $\widehat{sl_2}$ subalgebras for the the most general solution (\ref{eq:k2}) of the classical reflection equation (\ref{eq:re}). 
\vspace{1mm} 

From the point of view of applications to quantum integrable systems, the FRT presentation for the Onsager algebras and related current presentations
offer a different perspective for the analysis of well-known integrable models such as the Ising or superintegrable Potts model. 
For instance, by analogy with the analysis of \cite{BB3}, 
a free field realization of the Onsager's currents may be considered to solve the spectral problem for the mutually 
commuting quantities (\ref{eq:t1}) and study correlation functions of local fields. In particular, it would be nice if it could provide a derivation of the celebrated Painlev\'e equations \cite{Perk,MPW,FZ} for the Ising model based on the Onsager algebra.
\vspace{1mm}

Some of these problems will be considered elsewhere.

\vspace{1mm}

\vspace{0.5cm}
\noindent{\bf Acknowledgements:} We thank an anonymous referee  for comments.
We thank Hubert Saleur and Jasper Stokman for discussions, which motivated the construction of an FRT presentation for the Onsager algebra and its current presentation  that could extend to any higher rank cases.  
P.B.  and N.C. are supported by C.N.R.S.  S.B. and N.C. thanks the LMPT for hospitality, where part of this work has been done.
S.B. is supported by a public grant as part of the Investissement d'avenir project, reference ANR-11-LABX-0056-LMH,
LabEx LMH.

\end{document}